\documentclass[11pt]{article}
\usepackage[utf8]{inputenc}
\usepackage[T1]{fontenc}
\usepackage{lmodern}
\usepackage[DIV=11]{typearea}

\usepackage{graphicx}
\usepackage{amsmath,amssymb,amsthm}
\usepackage{complexity}
\usepackage{tikz}
\usepackage{subfigure}
\usepackage{algorithm}
\usepackage{algorithmicx}
\usepackage{algpseudocode}

\newtheorem{reduction}{Reduction}
\newtheorem{example}{Example}
\newtheorem{proposition}{Propositon}
\newtheorem{theorem}{Theorem}
\newtheorem{lemma}{Lemma}

\theoremstyle{definition}
\newtheorem{definition}{Definition}

\newcommand{\T}{\mathcal{T}}

\newcommand{\XX}{\mathcal{V}}
\newcommand{\OX}{\mathcal{O}}
\newcommand{\DX}{\mathcal{D}}
\newcommand{\PD}{\mathcal{PD}}
\DeclareMathOperator*{\vcl}{\operatorname{VE}}
\newcommand{\argmax}[1][0]{\underset{#1}{\arg\!\max}\;}
\newcommand*\samethanks[1][\value{footnote}]{\footnotemark[#1]}

\title{Maximizing a Submodular Function with Viability Constraints}
\author{Wolfgang Dvo\v{r}\'ak\thanks{Universit\"at Wien, Fakult\"at f\"ur Informatik}
\and Monika Henzinger\samethanks[1]
\and David P.\ Williamson\thanks{School of Operations Research and Information Engineering, Cornell University}
}
\date{}

\tikzstyle{root}=[]
\tikzstyle{species}=[draw, circle, inner sep=1.5pt]%[draw, thick, circle,inner sep=3pt, fill=gray!15]
\tikzstyle{selected}=[draw, circle, inner sep=1pt, fill=gray!15]

\renewcommand{\SC}{\mathcal{S}}

\begin{document}

\maketitle

\begin{abstract}
We study the problem of maximizing a monotone submodular function with viability constraints.
This problem originates from computational biology,
where we are given a phylogenetic tree over a set of species
and a directed graph, the so-called food web, encoding viability constraints between these species.
These food webs usually have constant {depth}.
The goal is to select a subset of $k$ species that satisfies the viability constraints
and has maximal phylogenetic diversity.
As this problem is known to be $\NP$-hard, we investigate approximation algorithms.
We present the first constant factor approximation algorithm if the depth is constant.
Its approximation ratio
is $(1-\frac{1}{\sqrt{e}})$.
This algorithm not only applies to phylogenetic trees with viability constraints
but for arbitrary monotone submodular set functions with viability constraints.
Second, we show that there is no $(1-1/e+\epsilon)$-approximation algorithm for our problem setting (even for additive functions) and
that there is no approximation algorithm for a slight extension of this setting.
\end{abstract}

\section{Introduction}\label{sec:intro}

We consider the problem of maximizing a monotone submodular set function $f$ over subsets of a ground set $X$, subject to a  restriction on what subsets are allowed. As discussed below, this problem has been well-studied with constraints on the allowed sets that are {\em downward-closed}; that is, if $S$ is allowed subset, then so is any $S' \subset S$.  Here we study the problem of maximizing such a function with a constraint that is not downward-closed.  Specifically, we assume that there exists a directed acyclic graph $D$ with the elements of $X$ as nodes in the graph
and only consider so-called \emph{viable} sets of a certain size.
A set $S$ is viable if each element either has no outgoing edges in $D$ or it has a path $P$ to such an element with $P \subseteq S$.
Such viability constraints are a natural way to model dependencies between elements, where an element can only contribute to the function if it appears together with specific other elements.

We are motivated by a  problem arising in conservation biology.  The problem is given as a rooted phylogenetic tree ${\cal T}$ with nonnegative weights on the edges, where the leaves of the tree represent species, and the weights represent genetic distance.  Given a conservation limit $k$, we would like to select $k$ species so as to maximize the overall phylogenetic diversity of the set, which is equivalent to maximizing the weight of the induced subtree on the $k$ selected leaves plus the root.  This problem,  known as
{\em Noah's Ark problem} \cite{Weizman98}, can be solved in polynomial time via a greedy algorithm \cite{Faith92,PardiG05,Steel05}.

When it comes the real-world instances the above formalization of Noah's Ark problem
has been criticized for not considering that the survival of a species might depend
on the survival of other species~(see e.g.~\cite{vanderHeide05}).
If one does not respect these dependencies a species might become extinct even if it is selected for preservation, 
which indeed would result in suboptimal solutions.
Moulton, Semple, and Steel~\cite{MoultonSS2007} introduced an extension of Noah's Ark problem 
which takes into account the dependence of various species on one another in a food web. 
In this food web, an arc is directed from species $a$ towards species $b$ if $a$'s survival depends on species $b$. Moulton et al.\ now consider selecting viable subsets given by the food web of size $k$, i.e.\ a
species is viable if also at least one of its successors in the food web is preserved.
Note that in real life the {\em depth} of the food web, i.e.,
the longest path between any node in $D$ and a node with no out-edge,
is rather small (usually no larger than 30) (see e.g. \cite{ChernomorMFKIHH15}). 

Faller et al.~\cite{FallerSW2007} show that the problem of maximizing phylogenetic diversity with viability constraints  is NP-hard, even in simple special cases with constant depth (e.g.\ the food web is a directed tree of constant depth).

Since phylogenetic diversity induces a monotone, submodular function on a set of species, this problem is a special case of the problem of  maximizing a submodular function with viability constraints.
There exists a long line of research on
approximately maximizing monotone submodular functions with constraints.  This line of work was initiated by Nemhauser et al. \cite{NemhauserWF78} in 1978; they give a greedy $(1 -\frac{1}{e})$-approximation algorithm for maximizing a monotone submodular function subject to a cardinality constraint.  Fisher et al. \cite{FisherNW78} introduced approximation algorithms for maximizing a monotone submodular function subject to matroid constraints (in which the set $S$ must be an independent set in a single or multiple matroids).  In recent work  other types of constraints have been studied, as well as nonmonotone submodular functions;  
see the surveys by Vondrak \cite{Vondrak13} and Goundan and Schulz~\cite{GoundanS07}.

In our case, the viability constraints are {\em not downward-closed} while most of the prior
work studies downward-closed constraints.
One notable exception, where not downward-closed constraints are considered, are matroid base constraints~\cite{LeeMNS09}.
The viability constraint could be extended to be downward-closed by simply 
defining every subset of a viable set to  be allowable.
However, this extension violates the exchange property of matroids and thus viability constraints also differ
from matroid base constraints. Hence we consider a new type of constraint in submodular function maximization. 
We show how variants on the standard greedy algorithm can be used to derive approximation algorithms for maximizing a monotone, submodular function with viability constraints; thus we show that a new type of constraint can be handled in submodular function maximization.

Specifically  we first present a scheme of $(1-\frac{1}{e^{p/(p+d-1)}})/2$ - approximation algorithms for monotone submodular set functions with viability constraints, 
where $d$ is the  minimum of the depth of the food web and $k$, and $p$ is a parameter of the algorithm.
Let $n$ be the number of species and $m$ the number of edges in the food web then above algorithm's running time is in 
$\OX\left( k \cdot ( 3^{p} n^{p+2} + n^{p+1}m)\right)$, i.e., the running time is exponential in $p$ but is polynomial for any fixed $p$.
For instance if we set $p=d$ we achieve a $(1-\frac{1}{\sqrt{e}})/2$ - approximation algorithm. %(recall that $d$ is bounded).

We further present a variant of these algorithms which are $(1-\frac{1}{e^{p/(p+d-1)}})$ - approximations,
but whose running time is $\OX\left( k \cdot ( 3^{p} n^{4p+3d-1} + n^{4p+3d-2}m)\right)$, i.e., exponential in both $d$ and $p$. 
For fixed $d\!=\!p$ this is polynomial and provides an $(1-\frac{1}{\sqrt{e}})$ - approximation algorithm.
However, as the running time heavily depends on $d$ and $p$ this is only feasible for small values of $d$.
For larger values of $d$ one has to find the right balance between running-time and approximation guarantee.

Next by a reduction from the maximum coverage problem, we show that there is no $(1-1/e+\epsilon)$-approximation algorithm for the phylogenetic diversity problem with viability constraints (unless $\P=\NP$).
However, the reduction from maximum coverage introduces a food web of linear depth. 
Thus, we further give a reduction from Max Vertex Cover that shows $\APX$-hardness even for instances with constant depth food webs.
Finally we consider a generalization of our problem where we additionally allow AND-constraints such as ``species $a$ is viable only  if we preserve \emph{both} species $b$ and species $c$'' and show that this generalization has no approximation algorithm (assuming $\P \neq \NP$) by a reduction from 3-SAT.

The remainder of the paper is structured as follows.
We define the problem more precisely in Section \ref{prob}, 
introduce our algorithms in Section \ref{alg}, 
and give the hardness results in Section \ref{hard}.

\section{Phylogenetic Diversity with Viability Constraints}
\label{prob}

We first give a formal definition of the problem.
\begin{definition}
A (rooted) \emph{phylogenetic tree} $\T=(T,E_\T)$ is a rooted tree with root $r$ and each non-leaf node having at least
2 child-nodes together with a weight function $w$ assigning non-negative integer weights to the edges.
Let $X_\T$ denote the set of leaf nodes of $\T$ also called species.
For any set $A \subseteq X_\T$ the operator $\T(A)$ yields the spanning tree of the set $A \cup \{r\}$,
and by  $\T_E(A)$ we denote the edges of this spanning tree.
Then for any set $S \subseteq X_\T$ the \emph{phylogenetic diversity} is defined as
$$
\PD(S)= \sum_{e: e \in \T_E(S)} w(e)
$$
A \emph{food web} $D$ for the phylogenetic tree $\T=(T,E_\T)$ is an acyclic directed graph $(X_\T,E)$.
A set $S \subseteq X_\T$ is called viable if each $s \in S$ is either a sink (a node with out-degree $0$) in $D$ or there is a $s' \in S$ such that $(s,s') \in E$.
\end{definition}

Below we exemplify these concepts.

\begin{example}\label{example1}
Consider the phylogenetic tree given in Figure~\ref{figure:example1}.
The root $r$ is at the top and the set of species $X=\{A,B,C,D,E\}$ at the bottom.
We omitted names for the two inner nodes.
Now we have that $\PD(\{A\})=2$, $\PD(\{B\})=3$ and $\PD(\{A,B\})=4$. 
Finally considering the concept of viable sets. 
Considering the given food web we have that 
$\{A,B\},\{A,B,D\}$ are viable as $B$ is a sink and 
for $A, D$ one of the successors is included in the set.
On the other hand side the sets $\{A\}$, $\{A,D,E\}$ are not viable as none of the successors of $A$ is included.
\begin{figure}[thb]
\centering
\subfigure[Phylogenetic Tree]{ \begin{tikzpicture}[scale=1,>=stealth]
		\path 	node[root](r){r}
			++(-1.5,-1)node[species](i1){}
			++(3,0)node[species](i2){}
			(-2,-2)node[species](a){A}
			++(1,0)node[species](b){B}
			++(1,0)node[species](c){C}
			++(1,0)node[species](d){D}
			++(1,0)node[species](e){E}
			;
		\scriptsize
		\path [-,] %
			(r) edge node[fill=white] {$1$}  (i1)
			(r) edge node[fill=white] {$2$}  (i2)
			(i1) edge node[fill=white] {$1$}  (a)
			(i1) edge node[fill=white, inner sep=3pt] {$2$} (b)
			(i2) edge node[fill=white, inner sep=3pt] {$2$} (d)
			(i2) edge node[fill=white, inner sep=3pt] {$1$} (e)
			(r) edge node[fill=white] {$3$} (c)
			 ;
		\end{tikzpicture}}
\hspace{30pt}
\subfigure[Food Web]{\begin{tikzpicture}[scale=1,>=stealth]
		    \path   (0,0) node[species](a){A}
			    ++(-0.5,-1)node[species](b){B}
			    ++(1,0)node[species](c){C}
			    (1.5,0)node[species](d){D}
			    ++(0.5,-1)node[species](e){E}
			    ;
		    \scriptsize
		    \path [->,] %
			    (a) edge (b)
			    (a) edge (c)
			    (d) edge (a)
			    (d) edge (e)
			    ;
		\end{tikzpicture}}
\caption{An illustration of Example~\ref{example1}.}
\label{figure:example1}
\end{figure}
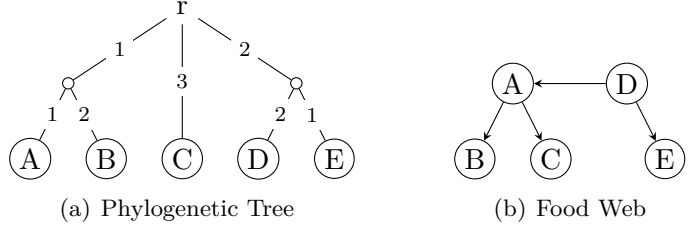
\end{example}

\pagebreak
Now our problem of interest is defined as follows.

\begin{definition}[OptPDVC]
The \emph{Optimizing Phylogenetic Diversity with Viability Constraints} (OptPDVC) problem is defined as follows.
You are given a phylogenetic tree $\T$ and a food web $D=(X_\T,E)$, and a positive integer $k$.
Find a viable subset $S \subseteq X_\T$ of size (at most) $k$ maximizing $\PD(S)$.
\end{definition}
OptPDVC is known to be $\NP$-hard \cite{FallerSW2007}, even for restricted classes of phylogenetic trees and dependency graphs.

First we study fundamental properties of the function $\PD$.
\begin{definition}
The set function $\PD(.|.): 2^X \times 2^X \mapsto \mathbb{N}_0$ for each $A,B \subseteq X$  is defined as
$\PD(A|B)=\PD(A \cup B) - \PD(B)$.
\end{definition}
The intuitive meaning of $\PD(A|B)$ is the gain of diversity we get by adding the set $A$ to the already selected species $B$.

We next recall the definition of submodular set functions.
We call a set function \emph{submodular} if
$$
\forall A, B, C \subseteq \Omega: A \subseteq B \Rightarrow f(A\cup C)-f(A)\geq f(B\cup C)-f(B).
$$
It turns out that $\PD$ as defined above happens to be a submodular function.
\begin{proposition}\label{prop:pd_submodular}
$\PD$ is a non-negative  monotone submodular function~\cite{bordewichS12}.
\end{proposition}

Now consider the function $\PD(.|.)$.
As $\PD(.)$ is monotone also $\PD(.|.)$ is monotone in the first argument and because of the submodularity of $\PD(.)$
the function $\PD(.|.)$ is anti-monotone in the second argument.

In the remainder of the paper we will not refer to the actual definition of the functions $\PD(.)$, $\PD(.|.)$,
but only exploit monotonicity, submodularity and the fact that these functions can be efficiently computed.

Moreover we will consider a function $\vcl$ (\emph{viable extension}) which, given a set of species $S$, returns a
viable set $S'$ of minimum size containing $S$. In the simplest case where $S$ consist of just one species it computes a shortest
path to any sink node in the food web.
Given a food web $D$ and the set $\mathcal{P}_D$ of paths that end in leaf node of $D$,
we define the \emph{truncated  depth}
\footnote{Notice that we use a slightly different definition for $d$ than in~\cite{DvorakHW13}.}
$d$ of $D$ as
$$
    d(D) = \min(\max_{P \in \mathcal{P}_D} |P|, k)
$$
i.e., as the minimum of the cardinality of the longest path that ends in a node without outgoing edges and $k$. 
If the food web is clear from the context we just write $d$ instead of $d(D)$. 
Note that the problem remains $\NP$-hard for instances with $d=2$, even if $\PD$ is additive~\cite{FallerSW2007}.
Finally we define the \emph{costs} $c(A|B)$ of adding a set of species $A$ to a set $B$
$$
c(A|B)= |\vcl(B \cup A)|-|B|.
$$
Notice that although $B$ typically will be a viable set this is not required in the definition.

\section{Approximation Algorithms}
\label{alg}
In this section we assume that a non-negative, monotone submodular function $\PD(.)$ is given as an oracle and we want to maximize  $\PD(.)$ under viability constraints together with a cardinality constraint.
We first review the greedy algorithm given by Faller et al.~\cite{FallerSW2007} presented in
Algorithm~\ref{alg:faller}.
The idea is that, in each step, one considers only species which either have no successors in the food web
or for which one of the successors has already been selected (adding one of these species will keep the set viable).
Then one adds the species that gives the largest gain of phylogenetic diversity.
\begin{algorithm}[t]
\caption{Greedy, Faller et al.}
\begin{algorithmic}[1]
\State $S \leftarrow \emptyset$%, $\BX \leftarrow X$
\While{$|S|<k$}
        \State $s \leftarrow \argmax[c(s|S)=1] v(\{s\} | S)$
	\State $S \leftarrow S \cup \{s\}$
\EndWhile
\end{algorithmic}
\label{alg:faller}
\end{algorithm}
By the restriction on the considered species the constructed set is always viable,
but we might miss highly valuable species which is illustrated by the following example.
\begin{example}\label{example2}
Consider the set of species $X=\{y,z,x_1,x_{2}\}$ the phylogenetic tree
$\T=(\{r\} \cup X,\{(r,y),(r,z),(r,x_1),(r,x_2)\}$,
weights $w(r,x_i)=1$, $w(r,y)=0$, $w(r,z)=C$ with $C>2$ and the food web $(X,\{(z,y)\})$ (see Fig.~\ref{figure:example2}).
Assume a budget $k=2$. As the species $y$ has weight $0$, Algorithm~\ref{alg:faller}
would pick $x_1,$ and $x_2$.
Hence Algorithm~\ref{alg:faller} results in a viable set with diversity $2$.
But the set $\{z,y\}$ is viable and has diversity $C$, which can be made arbitrarily large.
\begin{figure}[b]
\centering
\subfigure[Phylogenetic Tree]{ \begin{tikzpicture}[scale=1,>=stealth]
		\path 	node[root](r){$r$}
			++(-1.5,-1.5)node[species, inner sep=2pt](y){$y$}
			++(1,0)node[species, inner sep=2.5pt](z){$z$}
			++(1,0)node[species](x1){$x_1$}
			++(1,0)node[species](x2){$x_2$}
			;
		\path [-,] %
			(r) edge node[fill=white] {$0$}  (y)
			(r) edge node[fill=white, inner sep=2pt] {$C$} (z)
			(r) edge node[fill=white] {$1$} (x1)
			(r) edge node[fill=white] {$1$} (x2)
			 ;
\end{tikzpicture}}
\hspace{30pt}
\subfigure[Food Web]{ \begin{tikzpicture}[scale=1,>=stealth]
		\path 	node[species,inner sep=2pt](y){$y$}
			++(0,1)node[species, inner sep=2.5pt](z){$z$}
			++(1,-1)node[species](x1){$x_1$}
			++(1,0)node[species](x2){$x_2$}
			;
		\path [->]
	             (z) edge (y)
			;
\end{tikzpicture}}
\caption{An illustration of Example~\ref{example2}.}
\label{figure:example2}
\end{figure}
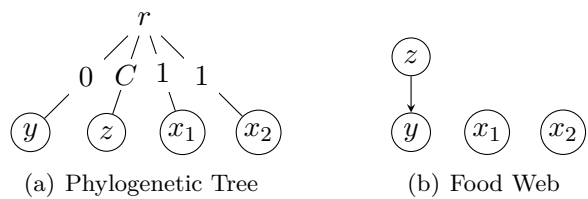
\end{example}
This example shows that the solution computed by the greedy algorithm can perform arbitrarily bad,
because it ignores highly weighted species if they are ``on the top of'' less valuable species.
Hence one can not give any approximation guarantee for Algorithm~\ref{alg:faller}.

\subsection{A Greedy Approximation Algorithm}\label{subsec:greedy}
By the above observations, to get an approximation guarantee, 
we have to consider all subsets of species up to a certain size $p$
which can be made viable and pick the most valuable subset.
Algorithm~\ref{alg:par_approx} exploits this observation.
It generalizes concepts from Bordewich and Semple~\cite{BordewichS2008}, 
which itself builds on Khuller et al.~\cite{KhullerMN99}.
In the algorithm $G$ denotes the current set of selected species; and
$\mathcal{G}$ denotes the best viable set we have found so far.
Lines 3 - 7  of the algorithm implements a greedy algorithm that in each step selects among the sets of size $\leq p$
with costs that are within the remaining budget the most ``cost efficient'' subset of species,
i.e.\ the subset $S$ of species that maximizes the ratio of the increase in PD over the cost of adding $S$,
and adds it to the solution.
But Algorithm~\ref{alg:par_approx} does not solely run the greedy algorithm, it first computes the set with maximal $\PD$ among
all sets of size $\leq p$ that can be made viable.
In certain cases this set is better than the viable set obtained by the greedy algorithm, 
a fact that we exploit in the proof of Theorem~\ref{thm:par_approx}.
\begin{algorithm}[tb]
\caption{}
\label{alg:par_approx}
\begin{algorithmic}[1]
      \State Select a set $S$ with $|S| \leq p$, $c(S|\emptyset) \leq k$ and maximal $\PD(S)$
      \State $\mathcal{G} = \vcl(S)$
      \State $G \leftarrow \emptyset$
      \While{$|G|<k$}
			\State Select a set $S$ with $|S| \leq p$, $c(S|G) \leq k-|G|$ and maximal ratio $\frac{\PD(S|G)}{c(S|G)}$
			\State $G = \vcl(G \cup S)$
      \EndWhile
      \If{$\PD(G)>\PD(\mathcal{G})$}
	  \State $\mathcal{G}\leftarrow G$
      \EndIf
\end{algorithmic}
\end{algorithm}

The next theorem will analyze the approximation ratio of Algorithm~\ref{alg:par_approx}.

\begin{theorem}\label{thm:par_approx}
For all integers $p \geq 1$
Algorithm~\ref{alg:par_approx} gives a $(1-\frac{1}{e^{p/(p+d-1)}})/2$ approximation.
\end{theorem}
As for $p \geq k$ the algorithm enumerates all possible solutions and thus is optimal in the following
analysis we will assume that $p \leq k$. 
To prove Theorem~\ref{thm:par_approx},
we introduce some notation.
First let $O \subseteq X$ denote the optimal solution.
We will consider a decomposition $\DX_O$ of $O$ into sets $O_1, \dots, O_{\lceil k/p \rceil}$ of size $\leq p$.
By decomposition we mean that
(i) $\bigcup_{i=1}^{\lceil k/p \rceil} O_{i} = O$ and
(ii) $O_i \cap O_j = \emptyset$ if $i\not=j$.
Moreover we require that $|\vcl(O_i)| \leq p+d-1$ and $\sum_i |\vcl(O_i)| \leq \frac{k}{p}(p+d-1)$.
Next we show that such a decomposition $\DX_O$ always exists.
\begin{lemma}\label{lemma:decomposition}
There exist $\lceil \frac{k}{p} \rceil$ many pairs $(O_1,B_1),\!\dots\!, (O_{\lceil \frac{k}{p} \rceil},B_{\lceil \frac{k}{p} \rceil})$ 
such that
$O\!=\!\bigcup_{1 \leq i \leq \lceil \frac{k}{p} \rceil} O_i$, $O_i \cup B_i$ is viable, $|O_i| \leq p$, $| B_i| \leq d-1$ and
$\sum_i |O_i \cup B_i| \leq \frac{k}{p}(p+d-1)$.
\end{lemma}
\begin{proof}
The optimal solution $O$ is a viable subset of size at most $k$.
Consider the reverse graph $G$ of $D$ projected on the set $O$ i.e.\   $G=(O, E^{-} \cap (O \times O))$,
and add an artificial root $\tau$ that has an edge to all roots of $G$ (the sinks of $D$).
By the definition of $d(D)$ we obtain that in G
the longest path starting in $\tau$ is at most of size $d+1$.

Start a depth-first-search in $\tau$ and initialize $O_1,B_1$ with empty sets.
Whenever the DFS removes a node from the stack, we add this node to the current set $O_i$, $i \geq 1$.
When $|O_i|=p$ then we add the nodes on the stack, except $\tau$, to the set $B_i$,
but do not change the stack itself.
Then we continue the DFS with the next pair $(O_{i+1}, B_{i+1})$, again initialized  by empty sets.
Eventually the DFS stops, then the stack is empty and thus  $(O_{\lceil \frac{k}{p} \rceil}, \emptyset)$
is the last pair.
Notice since the longest path starting in $\tau$ is at most of size $d+1$,
there are at most $d$ nodes on the stack, one being the root $\tau$ and hence $|B_i| \leq d-1$.
Since the DFS removes each node exactly once from the stack, all the sets $O_i \subseteq O$ are disjoint
and all except the last one are of size $p$.
Hence the DFS produces $\lceil \frac{k}{p} \rceil$ many sets $O_i$ satisfying  $O_i \cup B_i$ is viable, $|O_i| \leq p$ and $|B_i| \leq d-1$.
Finally as, by construction, $B_{\lceil \frac{k}{p} \rceil}= \emptyset$ and $|O_{\lceil \frac{k}{p} \rceil}|=k \mod p$ we obtain that
$\sum_{i=1}^{\lceil \frac{k}{p} \rceil} |O_i \cup B_i|= \sum_{i=1}^{\lfloor \frac{k}{p} \rfloor} |O_i \cup B_i| + (k \mod p)
\leq \lfloor \frac{k}{p} \rfloor(p+d-1) + (k \mod p) \leq \frac{k}{p}(p+d-1)$.
\end{proof}
\smallskip

First, we consider the greedy algorithm and the value $l$ where the $l$-th iteration is the first iteration such that
after executing the loop body,
\linebreak $\max_{O_j \in \DX_O}\frac{\PD(O_j| G)}{c(O_j|G)} > \max_{|S| \leq p, c(S|G) \leq k-|G|}\frac{\PD(S| G)}{c(S|G)}$.
If the greedy solution is different from the optimal one the inequality holds at least for the last iteration of the loop where $S=\emptyset$.
We define $O'_{l+1}=\argmax[O_j \in \DX_O]\frac{\PD(O_j| G)}{c(O_j|G)}$,
i.e.\ $O'_{l+1}$ is in the optimal viable set and would be a better choice than the selection of the algorithm,
but the greedy algorithm cannot make $G \cup O'_{l+1}$ viable without violating the cardinality constraint.

Let $S_i$ denote the set $S$ added to $G$ in iteration $i$ of the while loop in Line~4.
Moreover, for $i \leq l$ we denote the set $G$ after the $i$-th iteration by $G_i$, with $G_0=\emptyset$,
the set $G \cup S$ from Line 6 as $G^*_i=G_{i-1} \cup S_i$
and the ``costs'' of adding set $S_i$ by $c_i=c(S_i|G_{i-1})=c(G_i|G_{i-1})$.
With a slight abuse of notation we will use $G_{l+1}$ to denote the viable set $\vcl(G_{l} \cup O'_{l+1})$, $c_{l+1}$
to denote $c(O'_{l+1}|G_l)$ and
$G^*_{l+1}$ to denote the set $G_{l} \cup O'_{l+1}$ ($G_{l+1}$ is not a feasible solution as $|\vcl(G_{l+1})|>k$).
Notice that while the sets $G^*_i$ are not necessarily viable sets, all the $G_i, i \geq 0$ are viable sets and moreover
$\PD(\mathcal{G}) \geq \PD(G_i) , i \leq l$.

First we show that in each iteration  of the algorithm the set $S_i$ gives
a certain approximation of the missing part of the optimal solution.

\begin{lemma}\label{lemma:par1}For all integers $1 \leq p \leq k$ and $1 \leq i \leq l+1$:
$$
  \frac{\PD(G^*_i|G_{i-1})}{c_i} \geq \frac{p}{(p+d-1)k} \cdot \PD(O|G_{i-1})
$$
\end{lemma}
\begin{proof}
By definition of $S_i$ and $O'_{l+1}$ for each $O_j \in \DX_O$
the following holds:
{\small
$$
    \frac{\PD(O_j| G_{i-1})}{c(O_j|G_{i-1})} \leq \frac{\PD(S_i|G_{i-1})}{c(S_i|G_{i-1})}\ \text{for $i\leq l$}
    \quad \ \quad
    \frac{\PD(O_j| G_{l})}{c(O_j|G_{l})} \leq \frac{\PD(O'_{l+1}|G_{l})}{c(O'_{l+1}|G_{l})}
$$}\\
Combining the two inequalities we have 
{\small
$$
   \frac{\PD(O_j| G_{i-1})}{c(O_j|G_{i-1})} \leq \frac{\PD(G^*_i|G_{i-1})}{c(G_i|G_{i-1})}\quad \text{for $i\leq l+1$}
$$}\\
Next we use the monotonicity and submodularity of $\PD$ (for the first inequality) and
the inequality from above (for the second inequality).
{\small
\begin{align*}
\PD(O| G_{i-1}) \leq \sum_{O_j \in \DX_O} \PD(O_j|G_{i-1})
    = \sum_{O_j \in \DX_O} \frac{\PD(O_j|G_{i-1})}{c(O_j|G_{i-1})}\ c(O_j|G_{i-1})\\
    \leq \sum_{O_j \in \DX_O}  \frac{\PD(G^*_i|G_{i-1})}{c_i}\ c(O_j|G_{i-1})
    \leq \frac{\PD(G^*_i|G_{i-1})}{c_i} \ \frac{p+d-1}{p} \cdot k
\end{align*}
}
The last step exploits that by Lemma~\ref{lemma:decomposition} $\sum_{O_j \in \DX_O} c(O_j|G_{i-1}) \leq \frac{k}{p} \cdot(p+d-1)$.
\end{proof}

\begin{lemma}\label{lemma:par2}
For $1 \leq i \leq l+1$:
$$
  \PD(G^*_i) \geq
\left[
  1 -
  \prod_{j=1}^i
    \left(
	1 - \frac{p \cdot c_j}{(d+p-1)\cdot k}
    \right)
\right]
\  \PD(O)
$$
\end{lemma}
\begin{proof}
 The proof is by induction on $i$. The base case for $i=1$ is by Lemma~\ref{lemma:par1}.
 For the induction step we show that if the claim holds for all $i'< i$ then it must also hold for $i$.
 For convenience we define $C_i=  \frac{p \cdot c_i}{(d+p-1)\cdot k}$.% and again use Lemma~\ref{lemma:par1}.
{\small
\begin{align*}
   \PD(G^*_i) &=  \PD(G_{i-1}) +\PD(G^*_i|G_{i-1})
                \geq \PD(G_{i-1}) + C_i \cdot \PD(O | G_{i-1})\\
                &= \PD(G_{i-1}) + C_i \cdot \left( \PD(O \cup G_{i-1}) - \PD(G_{i-1})  \right)\\
		&\geq\left( 1 - C_i \right) \cdot \PD(G_{i-1}) +  C_i\cdot  \PD(O)\\
		&\geq \left( 1 - C_i\right) \cdot \PD(G^*_{i-1}) +  C_i\cdot  \PD(O)\\
		&\geq \left( 1 - C_i \right)
\left[
  1 -
  \prod_{j=1}^{i-1}
    \left(
	1 - C_j
    \right)
\right]
\  \PD(O)  +  C_i\cdot  \PD(O)\\
&=
\left[
  1 -
  \prod_{j=1}^i
    \left(
	1 - C_j
    \right)
\right]
\  \PD(O)\\
\quad& \qedhere
\end{align*}}
\end{proof}

\begin{proof}[Theorem~\ref{thm:par_approx}]
We first give a bound for $G^*_{l+1}$.
To this end consider $\sum_{m=1}^{l+1} c_m$.
As $G_{l+1}$ exceeds the cardinality constraint $\sum_{m=1}^{l+1} c_m > k$ it follows that:
{\small
\begin{align*}
1 -
\prod_{j=1}^{l+1}
  \left(
    1- \frac{p \cdot c_j}{(d+p-1)\cdot(k)}
  \right)
&\geq
1 -
\prod_{j=1}^{l+1}
  \left(
    1- \frac{p \cdot c_j}{(d+p-1)\cdot \sum_{m=1}^{l+1} c_m}
  \right)\\
\geq
1 -
\left(1-\frac{p}{(d+p-1) \cdot (l+1)}\right)^{l+1}
&\geq 1- \frac{1}{e^{p/(d+p-1)}}
\end{align*}}
To obtain second inequality we used the fact that the term
$\prod_{j=1}^{l+1}
  \left(
    1- \frac{c'_j}{C}
  \right)$
with constant $C$ and the constraint $\sum_{j=1}^{l+1} c'_j=1$ has its maximum at $c'_j=1/(l+1)$.

By Lemma~\ref{lemma:par2} we obtain
$
  \PD(G^*_{l+1}) \geq
\left(
    1- \frac{1}{e^{p/(d+p-1)}}
\right)
\cdot \PD(O)
$, thus it
only remains to relate $\PD(G^*_{l+1})$ to  $\PD(G_{l})$.
To this end we consider the optimal set of size $p$ selected in Line~2 and denote it by $S_o$.
If the greedy solution has higher $\PD$ than $S_o$ the algorithm returns a superset of $G^*_{l}$
otherwise a superset of $S_o$.
Hence, $\PD(G)$ is larger or equal to the maximum of $\PD(G^*_{l})$ and $\PD(S_o)$.
From the definitions of $G^*_{l}$ and $S_o$ it follows that
$$
\PD(G^*_{l+1}) \leq \PD(G^*_{l}) + \PD(O'_{l+1}) \leq \PD(G^*_{l}) + \PD(S_o).
$$
Now we have that either $\PD(G^*_{l}) \geq \PD(G^*_{l+1})/2 $ or $\PD(S_o) \geq \PD(G^*_{l+1})/2$ and 
thus obtain the following.
$$\PD(G) \geq \max(\PD(G^*_{l}),\PD(S_o)) \geq
\left(
    1- \frac{1}{e^{p/(d+p-1)}}
\right)
\cdot  \frac{\PD(O)}{2}$$
Hence Algorithm~\ref{alg:par_approx}  provides an $\left(
    1- \frac{1}{e^{p/(d+p-1)}}
\right)/2$ - approximation.
\end{proof}

\begin{theorem}\label{thm:par_runtime}
 Algorithm~\ref{alg:par_approx} runs in time $\OX\left( k \cdot ( 3^{p} n^{p+2} + n^{p+1}m)\right)$
 where $n$ is the number of species and $m$ the number of edges in the food web.
\end{theorem}

\begin{proof}
First notice that computing the function $\vcl$ %is basically a directed Steiner tree problem.
can be reduced
to a Steiner tree problem as follows. 
To compute $\vcl(S)$ first modify the graph $D$ by
(i) taking all the species in $S$ that are already connected (via nodes in $S$) to a sink node in $S$, and merging these nodes
into a single terminal node $t$, and
(ii) connecting the remaining sink nodes in $D$ to $t$.
For the starting nodes in the Steiner tree problem we use the species in $S$ that are not viable.
In  Algorithm~\ref{alg:par_approx} when computing $\vcl(G \cup S)$, there are at most $p$ starting nodes
as the set $G$ is made viable after each iteration and thus the viable set $G$ is contracted to the single vertex $t$.
The Steiner tree problem on acyclic directed graphs can be solved in time $\OX(3^j n^2 + nm)$ \cite{HsuTWL05}, where $j$ is the number of starting and terminal nodes.
In Line~1 we have to consider $\OX(n^{p})$ sets $S$ and for each of them we solve a Steiner tree problem.
with at most $p$ starting nodes. So this first loop can be done in time $\OX(3^{p} n^{p+2} + n^{p+1}m)$.
The number of iterations of the while loop is bounded by $k$ and in each iteration, in Line~4,
we have to solve $\OX(n^p)$ Steiner tree problems with at most $p$ starting nodes.
Now as each iteration takes time $\OX(3^{p} n^{p+2} + n^{p+1}m)$
we get a total running time of
$\OX(k \cdot ( 3^{p} n^{p+2} + n^{p+1}m))$.
\end{proof}

\subsection{An Approximation Algorithm using the Enumeration Technique}\label{subsec:enumeration}
In this Section we use a  modification of the enumeration technique as described in~\cite{KhullerMN99},
to get rid of the factor $1/2$ in the approximation ratio.
The idea is to consider all (viable) sets of a certain size and for each of them to run the greedy
algorithm starting with this set. 
Finally one chooses the best of the produced solutions.
These sets typically have to contain three objects of interest, 
in the case of the maximum coverage problem~\cite{KhullerMN99} 
(cf. Def.~\ref{def:max_coverage} below)
just three sets from the collection $\SC$ and thus the running time there is only increased by a cubic factor.
However, in our setting such an object of interest is a pair $(O_i, B_i)$, 
i.e.\ a set $O_i$ of size $\leq p$ and a set $B_i$ of size $<d$ making $O_i$ viable.
Thus these three objects result in a set of size of $3p+3d-3$,
increasing the running time by a factor of  $n^{3p+3d-3}$.
Algorithm~\ref{alg:slow_alg} gives a precise formulation of the modified algorithm.

\begin{theorem}\label{thm:slow_alg}
Algorithm~\ref{alg:slow_alg} is an $\left( 1- \frac{1}{e^{p/(d+p-1)}}\right)$ - approximation algorithm for OptPDVC which runs in time
% $\OX\left( k n^{3p+3d-3} \cdot ( 2^{p} n^{p+2} + n^{p+1}m)\right)$.
$\OX\left( k \cdot ( 3^{p} n^{4p+3d-1} + n^{4p+3d-2}m)\right)$,
where $n$ is the number of species, $m$ the number of edges in the food web and $p\!\geq\!1$ the parameter used in Algorithm~\ref{alg:slow_alg}.
\end{theorem}
\begin{algorithm}[tb]
\caption{}
\label{alg:slow_alg}
\begin{algorithmic}[1]
\State $\mathcal{G} \leftarrow \emptyset$
\For{each $G \subseteq X$, $G$ viable, $|G| \leq \min(3 p +3d -3,k)$}
      \While{$|G|<k$}
	\State Select a set $S$ with $|S| \leq p$, $c(S|G) \leq k-|G|$ and maximal ratio $\frac{\PD(S|G)}{c(S|G)}$
        \State $G = \vcl(G \cup S)$
      \EndWhile
      \If{$\PD(G)>\PD(\mathcal{G})$}
	  \State $\mathcal{G}\leftarrow G$
      \EndIf
\EndFor
\end{algorithmic}
\end{algorithm}
The proof of the above theorem is similar to the above analysis for Algorithm~\ref{alg:par_approx}.
Again let $O$ be the optimal viable set and $\DX_O$ a decomposition of $O$, given by Lemma~\ref{lemma:decomposition}.
We consider the set $G^*_0=O^*_1 \cup O^*_2 \cup O^*_3$ with $\{O^*_1,O^*_2,O^*_3\}\!\subseteq\!\DX_O$ such that
$\{O^*_1,O^*_2\} =\hspace{-7pt} \argmax[\{O_i,O_j\} \subseteq \DX_O] \hspace{-7pt} \PD(O_i \cup O_j)$
and $O^*_3 = \argmax[O_i \in \DX_O] \PD(O^*_1 \cup O^*_2 \cup O_i)$
and the viable extension $G_0=G^*_0 \cup B^*_1 \cup B^*_2 \cup B^*_3$.
At some point Algorithm~\ref{alg:slow_alg}  will consider $G_0$.
We consider this iteration of the for loop in Line~2 and consider the greedy algorithm.
To this end we use the same notation as in the proof of Theorem~\ref{thm:par_approx},
the only difference being the definition of the set $G_0$ above.
Recall that the value $l$ is defined such that the $l$-th iteration is the first iteration such that
after executing the loop body, $\max_{O_j \in \DX_O}\frac{\PD(O_j| G)}{c(O_j|G)} > \max_{|S| \leq p, c(S|G) \leq k-|G|}\frac{\PD(B| G)}{c(B|G)}$.

If $|G_0|=k$ then Line~2 of the algorithm enumerates all possible solutions and will eventually find the optimal solution.
So in this case there is no need to study the greedy algorithm. 
Hence in the remainder of this section we will assume that $|G_0|<k$. 

\begin{lemma}\label{lemma:slow1} For $1 \leq i \leq l+1$, $p\in \{1,\dots, k\}$:
$$
  \frac{\PD(G^*_i|G_{i-1})}{c_i} \geq \frac{p}{(p+d-1)(k-|G_0|)} \cdot \PD(O | G_{i-1})
$$
\end{lemma}
\begin{proof}
By definition of $S_i$ and $O'_{l+1}$ for each $O_j \in \DX_O$
the following holds:
{\small
$$
    \frac{\PD(O_j| G_{i-1})}{c(O_j|G_{i-1})} \leq \frac{\PD(S_i|G_{i-1})}{c(S_i|G_{i-1})}\ \text{for $i\leq l$}
    \quad \ \quad
    \frac{\PD(O_j| G_{l})}{c(O_j|G_{l})} \leq \frac{\PD(O'_{l+1}|G_{l})}{c(O'_{l+1}|G_{l})}
$$}\\
Combining the two inequalities we have 
{\small
$$
   \frac{\PD(O_j| G_{i-1})}{c(O_j|G_{i-1})} \leq \frac{\PD(G^*_i|G_{i-1})}{c(G_i|G_{i-1})}\quad \text{for $i\leq l+1$}
$$}

Next we use the monotonicity and submodularity of $\PD$ (for the first inequality),
the inequality from above (for the second inequality).
We use $\DX'_O$ to denote $\DX_O \setminus  \{O_j \subseteq G_{i-1}\}$.
{\small
\begin{align*}
\PD(O| G_{i-1}) \leq \sum_{O_j \in \DX'_O} \PD(O_j|G_{i-1})
    = \sum_{O_j \in \DX'_O} 
	    \frac{\PD(O_j|G_{i-1})}
		  {c(O_j|G_{i-1})}\ 
	     c(O_j|G_{i-1})\\
    \leq \sum_{O_j \in \DX'_O}  
	      \frac{\PD(G^*_i|G_{i-1})}
		   {c_i}\ 
	      c(O_j|G_{i-1})
    = \frac{\PD(G^*_i|G_{i-1})}{c_i}  \sum_{O_j \in \DX'_O} c(O_j|G_{i-1})\\
    \leq \frac{\PD(G^*_i|G_{i-1})}{c_i} \ \frac{p+d-1}{p} \cdot(k-|G_0|) % here  is the problem we  need  G_0 instead of G^*_0
\end{align*}
}

For the last step consider the modified food web graph where
(i) all nodes in $G_0$ are deleted  and
(ii) for each node that had an edge to a node in $G_0$ all outgoing edges are deleted,
i.e.\ the node is considered as a species that does not depend on any other species.
% makes all nodes that had edges to any node in $G_0$ to sinks, i.e deleting outgoing edges of these nodes.
Then we can apply Lemma~\ref{lemma:decomposition} to this graph and the set $O'=O\setminus G_0$
to obtain a decomposition satisfying $\sum_{O_j \in \DX'_O} c(O_j|G_{i-1}) \leq \frac{k-|G_0|}{p} \cdot(p+d-1)$. 
\end{proof}

\begin{lemma}\label{lemma:slow2}
For $1 \leq i \leq l+1$:
$$
  \PD(G^*_i|G_0) \geq
\left[
  1 -
  \prod_{j=1}^i
    \left(
	1 - \frac{p \cdot c_j}{(d+p-1)\cdot(k-|G_0|)}
    \right)
\right]
\  \PD(O | G_0)
$$
\end{lemma}
\begin{proof}
 The proof is by induction on $i$. The base case for $i=1$ is by Lemma~\ref{lemma:slow1}.
 For the induction step we show that if the claim holds for all $i'< i$ then it must also hold for $i$.
 For convenience we define $C_i=  \frac{p \cdot c_i}{(d+p-1)\cdot(k-|G_0|)}$.
{\small
\begin{align*}
   \PD(G^*_i|G_0) &=  \PD(G_{i-1}|G_0) +\PD(G^*_i|G_{i-1})\\
		&\geq \PD(G_{i-1}|G_0) + C_i \cdot \PD(O | G_{i-1})\\
		&= \PD(G_{i-1}|G_0) + C_i \cdot \left(\PD(O \cup G_{i-1})-\PD(G_{i-1})\right) \\
		&= \PD(G_{i-1}|G_0) + C_i \cdot \left( \PD(O \cup G_{i-1})-\PD(G_0) - \left(\PD(G_{i-1})-\PD(G_0)\right)\right)\\
		&= \PD(G_{i-1}|G_0) + C_i \cdot \left( \PD(O \cup G_{i-1}|G_0) - \PD(G_{i-1}|G_0)\right)\\
		&\geq \left( 1 - C_i \right) \cdot \PD(G_{i-1}|G_0) +  C_i\cdot  \PD(O | G_0) \\
		&\geq \left( 1 - C_i\right) \cdot \PD(G^*_{i-1}|G_0) +  C_i\cdot  \PD(O | G_0)\\
		&\geq \left( 1 - C_i \right)
\left[
  1 -
  \prod_{j=1}^{i-1}
    \left(
	1 - C_j
    \right)
\right]
\  \PD(O | G_0)  +  C_i\cdot  \PD(O | G_0)\\
&=
\left[
  1 -
  \prod_{j=1}^i
    \left(
	1 - C_j
    \right)
\right]
\  \PD(O | G_0)\\
\quad& \qedhere
\end{align*}}
\end{proof}

We are now prepared to prove the claimed approximation ratio.

\begin{proof}[Theorem~\ref{thm:slow_alg} approximation ratio]
We first give a bound for $G^*_{l+1}$.
To this end consider $\sum_{m=1}^{l+1} c_m$.
As $G^*_{l+1}$ exceeds the cardinality constraint $\sum_{m=1}^{l+1} c_m > k-|G_0|$ and hence:
{\small
\begin{align*}
1 -
\prod_{j=1}^{l+1}
  \left(
    1- \frac{p \cdot c_j}{(d+p-1)\cdot(k-|G_0|)}
  \right)
&\geq
1 -
\prod_{j=1}^{l+1}
  \left(
    1- \frac{p \cdot c_j}{(d+p-1)\cdot \sum_{m=1}^{l+1} c_m}
  \right)\\
\geq
1 -
\left(1-\frac{p}{(d+p-1) \cdot (l+1)}\right)^{l+1}
&\geq 1- \frac{1}{e^{p/(d+p-1)}}
\end{align*}}
To obtain the second inequality we used that the term
$\prod_{j=1}^{l+1}
  \left(
    1- \frac{c'_j}{C}
  \right)$
with constant $C$ and the constraint $\sum_{j=1}^{l+1} c'_j=1$ has its maximum at $c'_j=1/(l+1)$.

By Lemma~\ref{lemma:slow2} we obtain
$
  \PD(G^*_{l+1}|G_0) \geq
\left(
    1- \frac{1}{e^{p/(d+p-1)}}
\right)
\ \PD(O|G_{0})
$ and thus it
only remains to relate $G^*_{l+1}$ to  $G_{l}$.
First as $G_0=O^*_1 \cup O^*_2 \cup O^*_3$ and by the definition of $O^*_1, O^*_2$ we get
$\PD(O^*_3| O^*_1 \cup O^*_2) \leq \PD(G_0)/3$.
Next consider
{\small
\begin{align*}
 \PD(G^*_{l+1})- \PD(G_l)
		    &= \PD(O'_{l+1} | G_l)
		   \leq \PD(O'_{l+1}| O^*_1 \cup O^*_2)\\
		   &\leq \PD(O^*_3 |O^*_1 \cup O^*_2)
		   \leq \PD(G_0)/3
\end{align*}}
Finally, we can combine our results to obtain the claim:
{\small
\begin{align*}
 \PD(G) &\geq \PD(G_l)
	\geq \PD(G^*_{l+1}) - \PD(G_0)/3\\
	&= \PD(G^*_{l+1}|G_0) + \PD(G_0) - \PD(G_0)/3\\
	&\geq \left(1-\frac{1}{e^{p/(p+d-1)}}\right)\ \left(\PD(O)-\PD(G_{0})\right) + 2/3 \cdot \PD(G_0)\\
	&\geq \left(1-\frac{1}{e^{p/(p+d-1)}}\right)\cdot \PD(O)
\end{align*}}
The last inequality is by the fact that $1- \frac{1}{e^{p/(p+d-1)}} \leq 2/3$ for all $p,d \geq 1$.
\end{proof}

To conclude the proof of Theorem~\ref{thm:slow_alg} we finally consider the running time of Algorithm~\ref{alg:slow_alg}.

\begin{proof}[Theorem~\ref{thm:slow_alg} - running time]
As discussed in the proof of Theorem~\ref{thm:par_runtime} computing the function $\vcl$ is basically a Steiner tree problem
and can be solved in time $\tilde\OX(3^j n^2 + nm)$, where $j$ is the number of starting nodes.
In the algorithm we have to consider $\OX(n^{3p+3d-3})$ sets $S$ and for each of them we start the greedy algorithm.
The number of iterations of the while loop is bounded by $k$ and in each iteration, in Line 4,
we have to solve $\OX(n^p)$ Steiner tree problems with at most $p$ starting nodes.
As each iteration takes time $\OX(3^{p} n^{p+2} + n^{p+1}m)$
we get a total running time of
$\OX\left( n^{3p+3d-3} \cdot k \cdot ( 3^{p} n^{p+2} + n^{p+1}m)\right)=$\linebreak
$\OX\left( k \cdot (3^{p} n^{4p+3d-1} + n^{4p+3d-2}m)\right)$.
\end{proof}

\section{Impossibility Results}
\label{hard}
In this section we investigate upper bounds for the approximation guarantees %ratios
one can achieve with approximation algorithms for OptPDVC.
If we allow arbitrary monotone submodular functions in OptPDVC  it is easy to see that no $1-\frac{1}{e}+\epsilon$-approximation algorithm
exists (unless $\P=\NP$). This is immediate by the corresponding result for Max Coverage (with cardinality constraints).
In the following we investigate impossibility results that even hold if one considers phylogenetic diversity functions. 

In Section~\ref{hard1} we show that when considering viability constraints there is no $1-\frac{1}{e}+\epsilon$-approximation algorithm
for OptPDVC even if $\PD$ is additive/modular. However, the hardness proof requires food webs of linear depth.
  
Thus in Section~\ref{hard2} we investigate OptPDVC instances with constant depth food webs and show that maximizing the phylogenetic diversity
is $\APX$-hard.
 
Finally, in Section~\ref{hard3} we consider a straightforward generalization of viability constraints,
where we additionally allow AND-constraints such as ``species $a$ is viable only if we preserve \emph{both} species $b$ and species $c$'',
and show the inapproximability of the phylogenetic diversity under these constraints.
\subsection{$(1-1/e)$-Hardness for Additive Functions with Linear Depth Food Webs}
\label{hard1}

Towards our hardness result for additive functions, 
we first give a formal definition of the Max Coverage problem and recall the corresponding hardness result from the literature.

\begin{definition}\label{def:max_coverage}
The input to the \emph{Max Coverage} problem is a set of domain elements \linebreak $D=\{1,2,\dots,n\}$, together with non-negative integer weights $w_1,\dots w_n$, a collection
$\SC= \{ S_1,\dots S_m\}$ of subsets of $D$ and a positive integer $k$.
The goal is to find a set $\SC' \subseteq \SC$  of cardinality $k$ maximizing
$\sum_{i \in \bigcup_{S \in \SC'} S}w_i$.
\end{definition}

\begin{proposition}\label{prop:maxcoverage}
There is no $\alpha$-approximation algorithm for Max Coverage with $\alpha>1-\frac{1}{e}$ (unless $\P=\NP$)~\cite{Feige98,KhullerMN99}.
\end{proposition}

Below we give a reduction from the Max Coverage problem to OptPDVC, and then we show that it is approximation ratio preserving.
The idea is to first model an additive function via phylogenetic tree as follows: We consider each element $i \in D$ as species
which has an edge to the root with weight $w_i$. All the other species which we will use to build the appropriate
food web are also connected to the root but with weight $0$, i.e.\ the do not contribute to the value of the function.
In the food web we have to encode that we may only pick an element $i \in D$ if we also pick one of the set $S_j$ containing $i$.
This is done by 
introducing species $S_{j,1}$ to $S_{j,n}$ for each set $S_j$ and
connecting $i$ to the nodes $S_{j,n}$ for each $S_j$ with $i \in S_j$.
Finally, we guarantee that only $k$ of the nodes $S_{j,n}$ are selected for viable sets by 
putting each $S_{j,n}$ as the top element of a chain of $n$ species and setting the budget to $(k+1)\cdot n$.
\footnote{While it is in principle possible to select $k+1$ nodes $S_{j,n}$, 
these sets cannot select any $i \in D$ and thus have value $0$ and can be neglected.}  

\begin{reduction}\label{reduction:max_coverage}
Given an instance $(D,\SC,k)$ of the Max Coverage problem,
we build an instance of OptPDVC as follows (cf. Fig.~\ref{figure:reduction_maxcover})
\begin{align*}
 X &=  D \cup \{S_{i,j} \mid S_i \in \SC, 1 \leq j \leq n\}\\
 E &= \{(j,S_{i,n}) \mid j \in S_i\} \cup \{(S_{i,j+1},S_{i,j}) \mid 1 \leq i \leq m, 1 \leq j < n \}\\
 \T &= (\{r\} \cup  X, \{(r,s) \mid s \in X\}) \\
%  E_\T &= \{(r,s) \mid s \in X\}\\
 w_e &=
\begin{cases}
 w_i & \text{if } e=(r,i), i \in D \\
 0 & \text{otherwise}\\
\end{cases}\\
 k'&=(k+1)\cdot n
\end{align*}
\end{reduction}

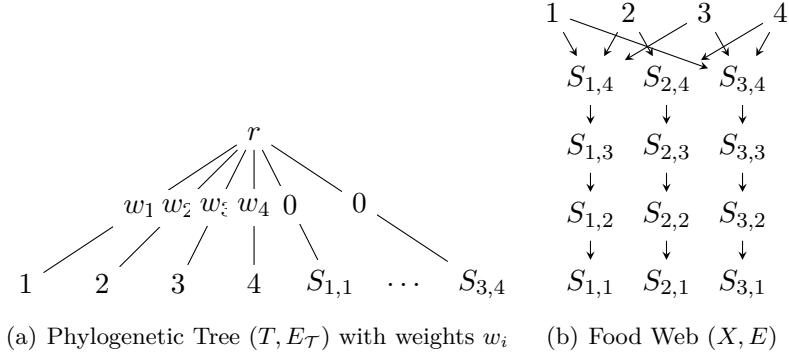
\begin{figure}[tb]
\centering
\tikzstyle{root}=[]
\tikzstyle{species}=[]
\subfigure[Phylogenetic Tree $(T,E_\T)$ with weights $w_i$]{ \begin{tikzpicture}[scale=1,>=stealth]
		\path 	node[root](r){$r$}
			++(-3,-2)node[species](1){$1$}
			++(1,0)node[species](2){$2$}
			++(1,0)node[species](3){$3$}
			++(1,0)node[species](4){$4$}
			++(1,0)node[species](s11){$S_{1,1}$}
			++(1,0)node[species](dots){$\dots$}
			++(1,0)node[species](smn){$S_{3,4}$}
			;
		\path [-,] %
			(r) edge node[fill=white] {$w_1$}  (1)
			(r) edge node[fill=white] {$w_2$} (2)
			(r) edge node[fill=white] {$w_3$} (3)
			(r) edge node[fill=white] {$w_4$} (4)
			(r) edge node[fill=white] {$0$} (s11)
% 			(r) edge node[fill=white] {$0$} (dots)
			(r) edge node[fill=white] {$0$} (smn)
			 ;
\end{tikzpicture}}
\subfigure[Food Web $(X,E)$]{ \begin{tikzpicture}[yscale=0.9,>=stealth]
		\path 	node[species](1){$1$}
			++(1,0)node[species](2){$2$}
			++(1,0)node[species](3){$3$}
			++(1,0)node[species](4){$4$}
			(0.5,-1)node[species](s14){$S_{1,4}$}
			++(1,0)node[species](s24){$S_{2,4}$}
			++(1,0)node[species](s34){$S_{3,4}$}
			(0.5,-2)node[species](s13){$S_{1,3}$}
			++(1,0)node[species](s23){$S_{2,3}$}
			++(1,0)node[species](s33){$S_{3,3}$}
			(0.5,-3)node[species](s12){$S_{1,2}$}
			++(1,0)node[species](s22){$S_{2,2}$}
			++(1,0)node[species](s32){$S_{3,2}$}
			(0.5,-4)node[species](s11){$S_{1,1}$}
			++(1,0)node[species](s21){$S_{2,1}$}
			++(1,0)node[species](s31){$S_{3,1}$}
			;
		\path [->,] %
			(1) edge (s14)
			(1) edge (s34)
			(2) edge (s14)
			(2) edge (s24)
			(3) edge (s14)
			(3) edge (s34)
			(4) edge (s24)
			(4) edge (s34)
			(s14) edge (s13)
			(s13) edge (s12)
			(s12) edge (s11)
			(s24) edge (s23)
			(s23) edge (s22)
			(s22) edge (s21)
			(s34) edge (s33)
			(s33) edge (s32)
			(s32) edge (s31)
			 ;
\end{tikzpicture}}
\caption{An illustration of Reduction~\ref{reduction:max_coverage}, applied to $D=\{1,2,3,4\}$, $\SC=\{S_1,S_2,S_3\}$,
$S_1=\{1,2,3\}$, $S_2=\{2,4\}$, $S_3=\{1,3,4\}$.}
\label{figure:reduction_maxcover}
\end{figure}

\begin{lemma}\label{lemma:max_coverage_correctness}
Let $(D,\SC,k)$ be an instance of the Max Coverage problem and let $(\T,(X,E),k')$ be the instance of OptPDVC given by Reduction~\ref{reduction:max_coverage}.
Let $W>0$.  Then there exists a cover $C \subseteq \SC$ of size $k$ with $w(C)\geq W$ for $(D,\SC,k)$
iff there exists a viable set $A$ of size $k' = (k+1)\cdot n$ with $\PD(A)\geq W$.
\end{lemma}
\begin{proof}
$\Rightarrow$: First  assume that there is a cover $C$ of size $k$  with $w(C)= W$.
We construct a viable set $A$ of size $k$ with $\PD(A)\geq W$ .
The set $A'=\{S_{i,j} \mid S_{i} \in C, 1 \leq j \leq n\} \cup \bigcup_{S_i \in C}S_i$ is a viable set of size $\leq k\cdot n + n$.
Clearly $\PD(A')= W$ and if $|A'|=k'$ we set $A=A'$ and are done.
If $|A'|<k'$ we can construct a viable set $A$ of size $k'$  with $\PD(A)\geq W$ by adding arbitrary viable species.

$\Leftarrow$: Assume there is a viable set $A$ of size $(k+1)\cdot n$  with $\PD(A) = W$.
We construct a cover $C$ of size $k$ with $w(C) \geq W$.
There are at most $k+1$ elements $S_i\in \SC$ such that $S_{i,n} \in A$.
This is by the fact that if $S_{i,n} \in A$ then also $S_{i,1}, \dots ,S_{i,n-1} \in A$.
Now consider the case where there are exactly $k+1$ such elements.
Then we already have  $(k+1)\cdot n$ species in $A$  and thus no $x \in D$ is contained in $A$.
But then $\PD(A) = 0$ as only the edges $(r,x)$ with $x\in D$ have non-zero weight.
Assuming $W > 0$ we thus have at most $k$ elements $S_i \in E$ such that $S_{i,n} \in A$
and further as $A$ is viable for each $x \in A$ there is an $S_{i,n} \in A$ such that $x \in S_i$.
Hence $C'=\{S_i \mid S_{i,n} \in A\}$ is of size at most $k$ and covers all $x \in A \cap D$,
i.e.\ $w(C') = W$. Now by adding arbitrary $S_i \in \SC$ we can construct a cover $C$ of size $k$ with $w(C) \geq W$.
\end{proof}

\begin{theorem}
There is no $\alpha$-approximation algorithm for OptPDVC with $\alpha>1-\frac{1}{e}$ (unless $\P=\NP$), even if $\PD$ is an additive function.
\end{theorem}
\begin{proof}
Immediate by Proposition~\ref{prop:maxcoverage}, Lemma~\ref{lemma:max_coverage_correctness} and the fact that Reduction~\ref{reduction:max_coverage} can be performed in polynomial time.
\end{proof}

Notice that in the above reduction the depth $d$ of the food web graph is not bounded by a constant and in fact is linear in $|X|$.

\subsection{$\APX$-Hardness for OptPDVC with Constant Depth Food Webs}
\label{hard2}

Using a result for Max Vertex Cover on bounded degree graphs, we can show $\APX$-hardness of OptPDVC with constant depth food webs.
\footnote{The authors are grateful to an anonymous reviewer who pointed them to the Max Vertex Cover problem.}

\begin{definition}\label{def:maxvertexcover}
The input to the \emph{Max Vertex Cover} is a graph $G=(V,E)$ (with bounded degree) together with a a positive integer $k$.
The goal is to find a set $S \subseteq V$ of cardinality $k$ that covers a maximum number of edges,
where an edge is covered if it is incident to at least one node in $S$.
\end{definition}

\begin{proposition}\label{prop:maxvertexcover}
Max Vertex Cover is $\APX$-hard for bounded degree graphs. 
In particular there is no PTAS unless $\P=\NP$~\cite{Patrank94}.
\end{proposition}

Below we give an reduction of Max Vertex Cover to our OptPDVC problem.
The main ideas are as follows (see also Fig.~\ref{figure:reduction_maxvertexcover}).
Given a graph $G=(V_G,E_G)$ with max. degree $\Gamma$,
for each node $v$ of the graph we make $\Gamma$ many copies $v_1,v_2, \dots, v_\Gamma$.
Then we build a phylogenetic tree with a root node $r$, inner nodes that correspond to the edges $E_G$
and the copies of nodes $v \in V_G$ as leaf nodes.
Each of the inner nodes is connected to the root via an edge of weight $1$
and as child-nodes it has one of the copies of each of the two  nodes incident to the corresponding edge in $G$.
Again the child-nodes are connected via an edge of weight $1$.
Moreover, each of the nodes $v_1,v_2, \dots, v_\Gamma$ is connected to at most one of the inner nodes and
those not connected to an inner node are connected to the root via an edge of weight $1$.

Then we build the food web graph of depth $\Gamma +1$ such that 
for each $v$ an optimal solution either picks all copies $v_1,v_2, \dots, v_\Gamma$ or none of them. 
To achieve this we have to introduce additional nodes, 
which we can add to the phylogenetic tree such that they do not
contribute to the phylogenetic diversity themselves, by setting the corresponding edge weights to $0$.

\begin{reduction}\label{reduction:maxvertexcover}
Given an instance $(G=(V_G,E_G),k)$ of the Max Vertex Cover,
with $\Gamma$ being the maximum degree of $G$.
For each node we assume an arbitrary order on the incident edges.
We build an instance of OptPDVC as follows 
\begin{align*}
 X &=  \{v_i, v'_i \mid v \in V_G, 1 \leq i \leq \Gamma\}\\
 E &= \{(v_i,v'_\Gamma) \mid v \in V_G, 1 \leq i \leq \Gamma \} \cup 
      \{(v'_i,v'_{i-1}) \mid v \in V_G, 2 \leq i \leq \Gamma \}\\
 \T &= (\{r\} \cup E_V \cup  X, E_\T)\\
  E_\T &= \{(r,f) \mid f \in E_G\} \cup  \{(r,v'_i) \mid v \in V_G, 1 \leq i \leq \Gamma\}\ \cup \\
	&\quad\ \{(f,v_i) \mid \text{$f$ is the $i$-th incident edge of v}\}\ \cup\\
	&\quad\ \{(r,v_i) \mid \text{$v$ has degree smaller than $i$}\}\\
 w_e &=
\begin{cases}
 1 & \text{if } e=(r,v_i), v \in V_G, 1 \leq i \leq \Gamma \\
 1 & \text{if } e=(f,v_i), f \in E, v \in V_G, 1 \leq i \leq \Gamma \\
 1 & \text{if } e=(r,f), f \in E \\
 0 & \text{otherwise}\\
\end{cases}\\
 k'&=2k\cdot \Gamma
\end{align*}
\end{reduction}

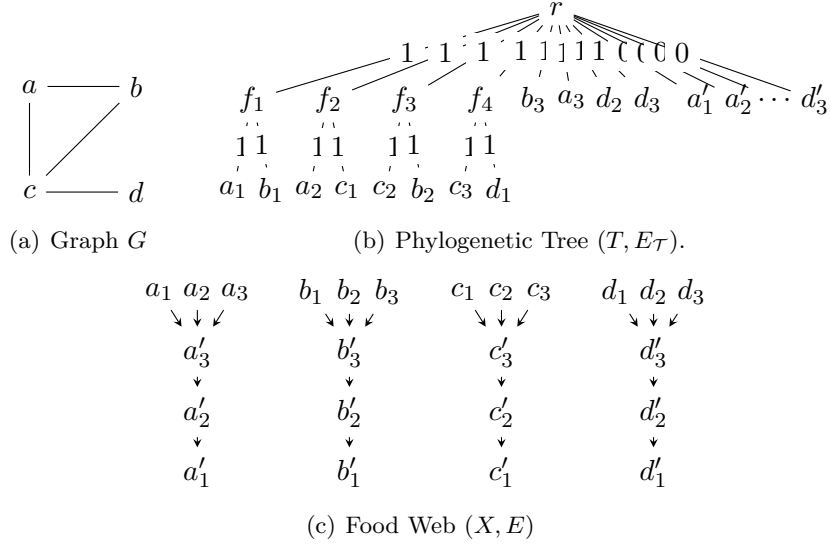
\begin{figure}[tb]
\centering
\tikzstyle{root}=[]
\tikzstyle{species}=[]
\subfigure[Graph $G$]{ \begin{tikzpicture}[scale=0.7,>=stealth]
		\path 	node[](a){$a$}
			++(2,0)node[](b){$b$}
			++(-2,-2)node[species](c){$c$}
			++(2,0)node[species](d){$d$}
			;
		\path [-,] %
			(a) edge (b)
			(a) edge (c)
			(b) edge (c)
			(c) edge (d)
			 ;
\end{tikzpicture}}
\hspace{10pt}
\subfigure[Phylogenetic Tree $(T,E_\T)$.]{ \begin{tikzpicture}[yscale=0.8,>=stealth]
		\path 	node[root](r){$r$}
			++(-4,-1.5)node[species](e1){$f_1$}
			++(1.0,0)node[species](e2){$f_2$}
			++(1.0,0)node[species](e3){$f_3$}
			++(1.0,0)node[species](e4){$f_4$}
			++(0.7,0)node[species](b3){$b_3$}
			++(0.5,0)node[species](a3){$a_3$}
			++(0.5,0)node[species](d2){$d_2$}
			++(0.5,0)node[species](d3){$d_3$}
			++(0.7,0)node[species](a1'){$a'_1$}
			++(0.5,0)node[species](a2'){$a'_2$}
			++(0.5,0)node[species](dots){$\dots$}
			++(0.5,0)node[species](d3'){$d'_3$}
			(-4.25,-3)node[species](a1){$a_1$}
			++(0.5,0)node[species](b1){$b_1$}
			++(0.5,0)node[species](a2){$a_2$}
			++(0.5,0)node[species](c1){$c_1$}
			++(0.5,0)node[species](c2){$c_2$}
			++(0.5,0)node[species](b2){$b_2$}
			++(0.5,0)node[species](c3){$c_3$}
			++(0.5,0)node[species](d1){$d_1$}
			;
		\path [-,] %
			(r) edge node[fill=white] {$1$} (e1)
			(r) edge node[fill=white] {$1$} (e2)
			(r) edge node[fill=white] {$1$} (e3)
			(r) edge node[fill=white] {$1$} (e4)
			(e1) edge node[fill=white] {$1$} (a1)
			(e1) edge node[fill=white] {$1$} (b1)
			(e2) edge node[fill=white] {$1$} (a2)
			(e2) edge node[fill=white] {$1$} (c1)
			(e3) edge node[fill=white] {$1$} (c2)
			(e3) edge node[fill=white] {$1$} (b2)
			(e4) edge node[fill=white] {$1$} (c3)
			(e4) edge node[fill=white] {$1$} (d1)
			(r) edge node[fill=white] {$1$} (b3)
			(r) edge node[fill=white] {$1$} (a3)
			(r) edge node[fill=white] {$1$} (d2)
			(r) edge node[fill=white] {$1$} (d3)
			(r) edge node[fill=white] {$0$} (a1')
			(r) edge node[fill=white] {$0$} (a2')
			(r) edge node[fill=white] {$0$} (dots)
			(r) edge node[fill=white] {$0$} (d3')
			 ;
\end{tikzpicture}}
\subfigure[Food Web $(X,E)$]{ \begin{tikzpicture}[yscale=0.8,>=stealth]
		\path 	(-3,0) node[species](a1){$a_1$}
			++(0.5,0)node[species](a2){$a_2$}
			++(0.5,0)node[species](a3){$a_3$}
			++(-0.5,-1)node[species](a3'){$a'_3$}
			++(-0,-1)node[species](a2'){$a'_2$}
			++(-0,-1)node[species](a1'){$a'_1$}
			(-1,0) node[species](b1){$b_1$}
			++(0.5,0)node[species](b2){$b_2$}
			++(0.5,0)node[species](b3){$b_3$}
			++(-0.5,-1)node[species](b3'){$b'_3$}
			++(-0,-1)node[species](b2'){$b'_2$}
			++(-0,-1)node[species](b1'){$b'_1$}
			(1,0) node[species](c1){$c_1$}
			++(0.5,0)node[species](c2){$c_2$}
			++(0.5,0)node[species](c3){$c_3$}
			++(-0.5,-1)node[species](c3'){$c'_3$}
			++(-0,-1)node[species](c2'){$c'_2$}
			++(-0,-1)node[species](c1'){$c'_1$}
			(3,0) node[species](d1){$d_1$}
			++(0.5,0)node[species](d2){$d_2$}
			++(0.5,0)node[species](d3){$d_3$}
			++(-0.5,-1)node[species](d3'){$d'_3$}
			++(-0,-1)node[species](d2'){$d'_2$}
			++(-0,-1)node[species](d1'){$d'_1$}
			;
			;
		\path [->,] %
			(a1) edge (a3')
			(a2) edge (a3')
			(a3) edge (a3')
			(a3') edge (a2')
			(a2') edge (a1')
			(b1) edge (b3')
			(b2) edge (b3')
			(b3) edge (b3')
			(b3') edge (b2')
			(b2') edge (b1')
			(c1) edge (c3')
			(c2) edge (c3')
			(c3) edge (c3')
			(c3') edge (c2')
			(c2') edge (c1')
			(d1) edge (d3')
			(d2) edge (d3')
			(d3) edge (d3')
			(d3') edge (d2')
			(d2') edge (d1')
			 ;
\end{tikzpicture}}
\caption{An illustration of Reduction~\ref{reduction:maxvertexcover}, applied to the Graph $G$ given in (a) with $\Delta=3$.}
\label{figure:reduction_maxvertexcover}
\end{figure}

\begin{lemma}\label{lemma:maxvertexcover_correctness}
Let $(G=(V_G,E_G),k)$ be an instance of Max Vertex Cover with maximal degree $\Gamma$, and let $(\T,(X,E),k')$ be the instance of OptPDVC given
by Reduction~\ref{reduction:maxvertexcover}.
There exists a vertex cover of size $k$ that covers $W$ many edges iff
there exists a viable set $A$ of size $k'$ with $\PD(A)\geq W + k \cdot \Gamma$.
\end{lemma}
\begin{proof}
$\Rightarrow$: First  assume that there is a vertex cover $S$ of size $k$ with $W$ many incident edges.
Now we can define a viable set $A=\{v_i, v'_i \mid v \in S, 1\leq i \leq \Gamma \}$.
It is easy to verify that $A$ is indeed viable and $|A|=2k \cdot \Gamma = k'$.
For $\PD(A)$ we first have the contributions of the edge in $\T$ that are incident to the nodes in $A$
which is $1$ for nodes $v_i$ and $0$ for nodes $v'_i$, in total $k \cdot \Gamma$.
Second, for each edge $(u,v)$ in $G$ that is covered by $S$ one of the corresponding nodes $u_i, v_{i'}$ is in $A$.
Thus in $\T$ the inner node $f$ correspond to $(u,v)$ has a child node in $A$ and 
the edge $(r,f)$ contributes $1$ to $\PD(A)$.
Hence, $\PD(A) = W + k \cdot \Gamma$.

$\Leftarrow$: Assume there is an optimal viable set $A$ of size $k'$  with $\PD(A) = W + k \cdot \Gamma$.
We first prove the following claim.

\emph{Claim:} For each $v \in V_G$ either all $\{v_i, v'_i \mid1\leq i \leq \Gamma \}$
are contained in $A$ or none of them.

Consider the sets $A_v= A \cap \{v_i, v'_i \mid1\leq i \leq \Gamma \}$ of nodes corresponding to $v$ and included in $A$.
As $A$ is a viable set, as soon as we have one $v_i$ in $A$ all $v'_i$ are in $A$.
Moreover, as $A$ is optimal also a species $v'_i$ is only in $A$ if at least one of the corresponding $v_i$ is in $A$.
Thus if $A_v$ is non-empty then $\Gamma+1 \leq |A_v| \leq 2 \Gamma$.
We have that (i) $\PD(A_v) \leq 2(|A_v|-\Gamma)$, as each $v_i$ can at most cover two edges of $\T$ (both with weight $1$)
and the $v'_i$ do not contribute to $\PD$ at all.
Towards a contradiction let us assume there is an $v \in A$ with (ii) $|A_v| < 2 \Gamma$.
As, by definition of $k'$, we can always find a viable set such that all $A_u, u \in V_G$ 
are either of size $2 \Gamma$ or empty, there must be a set $U\subset V_G$ of nodes
such that (a) $A_u\not=\emptyset$ for all $u \in U$
and (b) in total at least $|A_v|$ many nodes of $\bigcup_{u \in U }\{u_i  \mid 1\leq i \leq \Gamma \}$ 
are not included in $\bigcup_{u \in U} A_u$ (but all nodes $\bigcup_{u \in U }\{u'_i  \mid1\leq i \leq \Gamma \}$ are included).
That is we can consider $A \setminus A_v$ and add $|A_v|$ many of these uncovered nodes while maintaining viability.
While by excluding $A_v$ the diversity drops by less than $|A_v|$ (by (ii) $|A_v| < 2 \Gamma$ and thus by (i) $\PD(A_v) <|A_v|$)
for each node added we increase the diversity by at least $1$.
In total we increased the diversity, while maintaining viability, a contradiction to the optimality of $A$ 
which concludes the proof of the claim.

Given the claim we can define a set cover $S=\{v \mid v \in V_G, v_1 \in A\}$ that is of size $k$.
Now consider $\PD(A)$, which is composed of two parts.
First, the contribution from the edges (in $\T$) that are incident with nodes in $A$.
By the structure of $A$ (given in the Claim) there are $\Gamma \cdot k$ many nodes $v_i$ in A 
and each of them contributes $1$. Thus their total contribution is $\Gamma \cdot k$.
Second, the contribution by the edges (in $\T$) from the root to some inner node which has a child in $A$.
As $\PD(A) = W + k \cdot \Gamma$ these inner nodes contribute $W$ and as each of them contributes $1$ there are $W$ many such nodes.
We next show that $S$ covers at least $W$ edges.
Consider an inner node $f$ and the corresponding edge in $G$. 
We know that in $\T$ one child $u_i$ of $f$ is included in $A$ but then $u \in S$ and by construction
$u$ is incident to the edge $f$, i.e.\ the edge $f$ is covered by $S$.
Hence, $S$ covers at least $W$ many edges.
\end{proof}

\begin{theorem}
OptPDVC is $\APX$-hard for constant depth food webs. 
In particular there is no PTAS unless $\P=\NP$.
\end{theorem}
\begin{proof}
Immediate by Proposition~\ref{prop:maxvertexcover}, 
Lemma~\ref{lemma:maxvertexcover_correctness} and the fact that Reduction~\ref{reduction:maxvertexcover} can be performed in polynomial time.
Also notice that in the condition $\PD(A)\geq W + k \cdot \Gamma$ in Lemma~\ref{lemma:maxvertexcover_correctness} 
(i) $\Gamma$ is a constant as we consider bounded degree graphs for Max Vertex Cover and 
(ii) that for an optimal vertex cover/viable set $W$ is at least $k$. i.e., we can cover at least $k$ edges.
Thus, each $(1-\epsilon)$ approximation (resp.\ each PTAS) for OptPDVC would give a $(1-\Gamma \cdot \epsilon)$ approximation (resp.\ a PTAS) for 
Max Vertex Cover.
\end{proof}

\subsection{Inapproximability of OptPDVC with Generalized Viability Constraints}
\label{hard3}
Finally let us consider a straightforward generalization of viability constraints. %the food web.
So far we assumed that a species is viable iff at least one of its successors survives,
but one can also imagine cases where one node needs several or even all of its successors to survive to be viable.
In the following we consider food webs where we allow two types of nodes: (i) nodes that are viable if at least one successors survives
and (ii) nodes that are viable only if all successors survive.

We will show that in this setting no approximation algorithm  is possible using a
reduction from the $\NP$-hard problem of deciding whether a propositional formula in 3-CNF is satisfiable.
A 3-CNF formula is a propositional formula which is the conjunction of clauses, and each clause is the disjunction of exactly three literals,
e.g.\ $\varphi=(x_1 \vee x_2 \vee x_3) \wedge (x_2 \vee \neg x_3 \vee \neg x_4) \wedge (x_2 \vee x_3 \vee x_4)$.
The main idea behind Reduction~\ref{red:sat} is that we can build a food web such that a specific species is in a viable set iff 
the given formula is satisfiable.

\begin{reduction}\label{red:sat}
Given a propositional formula $\varphi$ in 3-CNF over propositional variables $\XX=\{x_1,\dots, x_n\}$ with clauses $c_1, \dots, c_m$ build the following instance
$(T,E_\T)$, $(X,E)$ and weight $w_e$ (cf. Fig.~\ref{figure:reduction_sat}) :
\begin{align*}
 X &= \{c_1,\dots, c_m\} \cup \{x, \bar{x}, c_x \mid x \in \XX\} \cup \{t\}\\
 \T &= (\{r\} \cup  X, \{(r,s) \mid s \in X\})\\
 w_e &=
\begin{cases}
 1 & e=(r,t) \\
 0 & \text{otherwise}\\
\end{cases}\\
 E &= \{(c_x,x),(c_x,\bar{x})\mid x \in \XX\} \cup \{(c_i,x) \mid x \in c_i\} \cup \{(c_i, \bar{x}) \mid \neg x \in c_i\}\\
   &\quad \cup \{(t,c_i),(t,c_x) \mid 1\leq i \leq m, x \in \XX\}\\
 k &= 2\cdot|\XX|+m+1
\end{align*}
The species $\{c_1,\dots c_m\} \cup \{x, \bar{x}, c_x \mid x \in \XX\}$ are viable in the traditional sense and $t$ is viable iff all its successors survive.
More formally, a set $S \subseteq X$ is viable if (i) for each $s \in S$ either $s$ is a sink or there is a $s' \in S$ with $(s,s') \in E$
and (ii) if $t \in S$ it holds for all $s'$ with $(t,s') \in E$ that $s' \in S$.
\end{reduction}

\begin{figure}[t!]
\centering
\tikzstyle{root}=[]
\tikzstyle{species}=[]
\subfigure[Phylogenetic Tree $(T,E_\T)$ with weights $w_i$]{ \begin{tikzpicture}[scale=0.73,>=stealth]
		\path 	node[root](r){$r$}
			++(-4,-2)node[species](t){$t$}
			++(1,0)node[species](x1){$x_1$}
			++(1,0)node[species](nx1){$\bar x_1$}
			++(1,0)node[species](cx1){$c_{x_1}$}
			++(0.7,0)node[species](dots1){$\dots$}
			++(0.7,0)node[species](x4){$x_4$}
			++(1,0)node[species](nx4){$\bar x_4$}
			++(1,0)node[species](cx4){$c_{x_4}$}
			++(1,0)node[species](c1){$c_1$}
			++(1,0)node[species](c2){$c_2$}
			++(1,0)node[species](c3){$c_3$}
			;
		\path [-,] %
			(r) edge node[fill=white, inner sep=0.3pt] {$1$}  (t)
			(r) edge node[fill=white, inner sep=0.3pt] {$0$} (x1)
			(r) edge node[fill=white, inner sep=0.3pt] {$0$} (nx1)
			(r) edge node[fill=white, inner sep=0.3pt] {$0$} (cx1)
			(r) edge node[fill=white, inner sep=0.3pt] {$0$} (x4)
			(r) edge node[fill=white, inner sep=0.3pt] {$0$} (nx4)
			(r) edge node[fill=white, inner sep=0.3pt] {$0$} (cx4)
			(r) edge node[fill=white, inner sep=0.3pt] {$0$} (c1)
			(r) edge node[fill=white, inner sep=0.3pt] {$0$} (c2)
			(r) edge node[fill=white, inner sep=0.3pt] {$0$} (c3)
			 ;
\end{tikzpicture}}
\subfigure[Food Web $(X,E)$]{
\begin{tikzpicture}[xscale=0.47,yscale=1,>=stealth]
 		\tikzstyle{arg}=[]
		\path 	node[arg](phi){$t$}
			++(-3,-1) node[arg](c1){$c_1$}
			++(3,0) node[arg](c2){$c_2$}
			++(3,0) node[arg](c3){$c_3$};
		\path 	(-3.8,-2.3)  node[arg](z1){$x_1$}
			++(1,0) node[arg](nz1){$\bar x_1$}
			++(1.2,0) node[arg](z2){$x_2$}
			++(1,0) node[arg](nz2){$\bar x_2$}
			++(1.2,0) node[arg](z3){$x_3$}
			++(1,0) node[arg](nz3){$\bar x_3$}
			++(1.2,0) node[arg](z4){$x_4$}
			++(1,0) node[arg](nz4){$\bar x_4$};
		\path 	(-4,-0.8) node[arg](cx1){$c_{x_1}$}
			++(2.7,0) node[arg](cx2){$c_{x_2}$}
			++(2.6,0) node[arg](cx3){$c_{x_3}$}
			++(2.5,0) node[arg](cx4){$c_{x_4}$};
		\path [left,<-]
			(c1) edge (phi)
			(c2) edge (phi)
			(c3) edge (phi);
		\path [left,<-]
			(cx1) edge (phi)
			(cx2) edge (phi)
			(cx3) edge (phi)
			(cx4) edge (phi)	
			;
		\path [left,<-,]
			(z1) edge (c1)
			(z2) edge (c1)
			(z2) edge (c2)
			(nz3) edge (c2)
			(z3) edge (c3)
			(z4) edge (c3)
			(z3) edge (c1)
			(z2) edge (c3)
			(nz4) edge (c2);

		\path [left,<-,]
			(z1) edge (cx1)
			(nz1) edge (cx1)
			(z2) edge (cx2)
			(nz2) edge (cx2)
			(z3) edge (cx3)
			(nz3) edge (cx3)
			(z4) edge (cx4)
			(nz4) edge (cx4);
\end{tikzpicture}}
\caption{An illustration of Reduction~\ref{red:sat}, applied to the propositional formula $\varphi=(x_1 \vee x_2 \vee x_3) \wedge (x_2 \vee \neg x_3 \vee \neg x_4) \wedge (x_2 \vee x_3 \vee x_4)$.}
\label{figure:reduction_sat}
\end{figure}
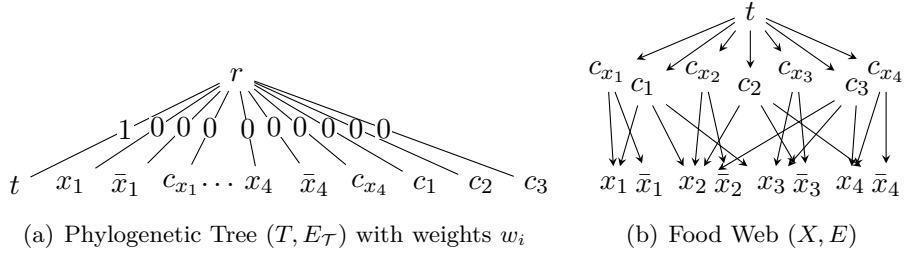

\begin{lemma}\label{lemma:sat_correctness}
Given a propositional formula $\varphi$ and the instance $(\TC,(X,E),k)$ of OptPDVC given by Reduction~\ref{red:sat}.
Then $\varphi$ is satisfiable iff there exists a viable set $A$ of size $\leq k$  with $\PD(A)> 0$.
\end{lemma}

\begin{proof}
$\Rightarrow:$
Let $\alpha$ be a truth assignment satisfying $\varphi$, i.e.\ $\alpha(\varphi)=1$.
Then it is easy to verify that
$A=\{x \mid x\in \XX, \alpha(x)=1\} \cup \{\bar{x} \mid x\in \XX, \alpha(x)=0\}
\cup \{c_1,\dots c_m\} \cup
\{c_x \mid x \in \XX \} \cup
\{t\}$ is a viable set of size $k=2\cdot|\XX|+m+1$ with $\PD(A) = 1$.

$\Leftarrow:$
If there is a viable subset $A'$ with $\PD(A') > 0$ there is also
a viable set $A \supset A'$ of size $k=2\cdot|\XX|+m+1$  and $\PD(A)> 0$, because
$|X|$ is of size $3\cdot|\XX|+m+1$.
We show that the truth-assignment $\alpha$ setting each $s \in A \cap \XX$ to $1$ and
each $s \in \XX \setminus A$ to $0$ satisfies $\varphi$.
As $\PD(A) > 0$ we clearly have that $t\in A$.
Now as $A$ is viable and we have an AND constraint on $t$ also $\{c_1,\dots c_m\} \cup \{c_x \mid x \in \XX \} \subseteq A$.
By $c_x \in A$ we obtain that for each $x \in \XX$ either $x \in A$ or $\bar{x} \in A$,
but not both of them (as the budget only allows $|\XX|$ further species).
Finally as $c_i \in A$ we have that for each clause there is either an $x \in C$  with $\alpha(x)=1$ or a
$\neg x \in C$ with  $\alpha(x)=0$. Thus each clause $c_i$ is satisfied by $\alpha$, i.e.\ $\alpha(c_i)=1$, and hence also $\alpha(\varphi)=1$.
\end{proof}

Now assuming that there is an approximation algorithm for OptPDVC with generalized viability constraints we would immediately get a procedure deciding 3-CNF formulas: apply Reduction~\ref{red:sat} to the formula, compute $\PD$ using the $\alpha$-approximation algorithm,
and return satisfiable if $\PD$ is positive.

\begin{theorem}
It is $\NP$-hard to decide whether an instance of OptPDVC with generalized viability constraints has a viable set $S$ with $\PD(S)>0$.
Thus no approximation algorithm for the problem can exist unless $\P = \NP$.
\end{theorem}
\begin{proof}
Immediate by  Lemma~\ref{lemma:sat_correctness},  and the fact that Reduction~\ref{red:sat} can be performed in polynomial time.
\end{proof}

\section*{Acknowledgments}
A preliminary version of this paper has been presented at 
the 21st European Symposium on Algorithms (ESA'13)~\cite{DvorakHW13}.

The research leading to these results has received funding from the European Research Council 
under the European Union's Seventh Framework Programme (FP/2007-2013) / ERC Grant
Agreement no.\ 340506.

\end{document}